\documentclass[journal,onecolumn]{IEEEtran}
\usepackage{amsmath,amssymb,amsthm,amsfonts}
\usepackage{mathtools}
\usepackage{bm}
\usepackage{hyperref}
\usepackage{enumitem}
\usepackage{mathrsfs}
\usepackage{orcidlink}

\newtheorem{theorem}{Theorem}

\newtheorem{remark}{Remark}

\newtheorem{proposition}{Proposition}

\ifCLASSINFOpdf
\else
\fi
\begin{document}
%
\title{Superdirectivity as Boundary Concentration under Spectral Collision}
%
%
%

\author{Hong~Yang\orcidlink{0000-0002-3744-7688}
\thanks{H. Yang (retired) was with the Department
of Mathematics and Algorithms Research, Bell Laboratories, Murray Hill,
NJ, 07974 USA e-mail: hyang.bell.labs@gmail.com (https://www.bell-labs.com/about/researcher-profiles/hyang/).}
\thanks{}}

\maketitle

\begin{abstract}
Array superdirectivity is traditionally derived through singular optimization of densely spaced antenna arrays. In this paper, we show that the phenomenon admits a geometric interpretation as a concentration effect induced by spectral collision. As the spacing of an $M$-element linear array tends to zero, the exponential family generated by a linear array undergoes a spectral collision, and the associated finite-dimensional subspaces converge in reproducing kernel to a polynomial jet space. The maximum achievable array gain equals the diagonal evaluation of the reproducing kernel, and is therefore governed by the reciprocal Christoffel function. For the classical flat \(L^2([-1,1])\) geometry, the Christoffel--Darboux kernel exhibits boundary concentration, yielding the quadratic \(M^2\) superdirective law as a direct consequence of kernel asymptotics. This viewpoint separates intrinsic gain limits from numerical conditioning and identifies superdirectivity as a manifestation of a more general concentration mechanism. The framework further shows that the classical \(M^2\) scaling is not universal: alternative spectral geometries produce different concentration laws through their associated Christoffel asymptotics. The results establish a direct connection between superdirectivity, reproducing kernels, orthogonal polynomials, and concentration phenomena arising from singular spectral limits.
\end{abstract}


%
\IEEEpeerreviewmaketitle

\section{Introduction}
%
%
%
%

Superdirectivity is a classical phenomenon in array theory in which a densely spaced array achieves directional gains far exceeding the $O(M)$ scaling of conventional $M$-element linear arrays. For a linear array operated in the endfire regime, the maximal directivity is known to scale asymptotically as $M^2$. This remarkable amplification law first appeared in the pioneering work of Uzkov~\cite{Uzkov1946}, who derived the quadratic law via quadratic form eigenvalue optimization of array gain. Schelkunoff~\cite{Schelkunoff1943} had earlier introduced the fundamental polynomial representation of linear arrays, mapping the array factor to a polynomial evaluated on the unit circle, but did not pursue the dense-spacing limit. Subsequent contributions include Parsons' proof~\cite{Parsons1987} that the linear endfire array is optimal among all three-dimensional configurations with vanishingly small separations. These classical derivations rely either on precise cancellation from asymptotic analysis of array factors, near-singular optimization, or on physical arguments based on derivative fields and strongly oscillatory current distributions. Despite these developments, the mathematical mechanism responsible for the quadratic scaling remains somewhat hidden behind the singular nature of the dense-array limit.

The purpose of this paper is to show that superdirectivity can be understood as a concentration phenomenon generated by a singular spectral limit. As the spacing of a linear array tends to zero, the exponential steering functions that describe the array manifold lose spectral separation and collapse onto a finite-dimensional polynomial jet space. We refer to this degeneration as a spectral collision. The resulting limit converts the superdirective gain problem into an extremal problem in a finite-dimensional reproducing-kernel Hilbert space (RKHS). Within this framework, the maximum achievable gain is represented by the diagonal of the reproducing kernel. The dense-array limit reduces the study of superdirectivity to the asymptotic behavior of reproducing kernels associated with polynomial spaces. For the classical flat \(L^2([-1,1])\) geometry, the limiting kernel is the Christoffel--Darboux kernel of the Legendre polynomials. The familiar \(M^2\) endfire law then follows directly from its endpoint asymptotics. This perspective reveals that superdirectivity is fundamentally a boundary concentration phenomenon. The reciprocal Christoffel function decays at a faster rate at the endpoints of the interval than in the interior, producing a concentration of kernel mass at the boundary. A trace identity fixes the total integrated gain, so the large endpoint amplification must be accompanied by a corresponding depletion of gain elsewhere. In this sense, superdirectivity is not the creation of additional gain, but rather a redistribution of a fixed gain budget induced by spectral concentration. The reproducing-kernel viewpoint also clarifies the role of geometry. The quadratic \(M^2\) law is not a universal consequence of dense spacing itself; rather, it reflects the particular spectral geometry associated with the flat Lebesgue measure on \([-1,1]\). When the underlying measure is changed, the corresponding Christoffel asymptotics yield different concentration laws. Thus superdirectivity appears as one instance of a broader class of concentration phenomena arising from singular spectral limits.

The analysis combines finite-dimensional RKHS theory, confluent limits of exponential families, and classical asymptotics of orthogonal polynomials. The resulting framework provides a geometric interpretation of superdirectivity, connects array gain to Christoffel functions and reproducing kernels, and identifies spectral collision as the mechanism underlying the classical endfire amplification law.

The remainder of the paper is organized as follows. Section \ref{s:pf} formulates the array-gain problem and introduces the associated exponential subspaces. Section \ref{s:ag} identifies array gain with diagonal evaluation of the reproducing kernel. Section \ref{s:scl} establishes the spectral-collision limit and the emergence of the polynomial jet space. Section \ref{s:super} derives superdirectivity from the limiting Christoffel--Darboux kernel. Section \ref{s:gsc} develops the concentration interpretation through Christoffel functions, trace identities, and alternative spectral geometries. Concluding remarks are given in Section \ref{s:concl}.

\section{Problem Formulation}\label{s:pf}
The ambient Hilbert space throughout this paper is  
\[
L^2([-1,1])
=
\left\{
f:[-1,1]\to\mathbb C
\;\middle|\;
\int_{-1}^1 |f(u)|^2\,du < \infty
\right\},
\]
the space of square-integrable complex-valued functions on the interval $[-1,1]$ with the inner product
\[
\langle f,g\rangle
=
\int_{-1}^1 f(u)\overline{g(u)}\,du,
\]
and induced norm
\[
\|f\|_{L^2}
=
\left(
\int_{-1}^1 |f(u)|^2\,du
\right)^{1/2}.
\]

\subsection{Linear Array}
Consider an $M$-element linear array with element locations
\[
x_n=d\alpha_n, \quad n=0,\cdots, M-1
\]
where \(d>0\) is a global scaling parameter for spacing and \(\alpha_n=\frac{x_n}{d}\in\mathbb R\) are fixed normalized dimensionless coordinates that encode the relative array geometry, and
\[
\delta = kd = \frac{2\pi}{\lambda}d
\]
is the normalized spacing parameter, with \(k=2\pi/\lambda\) the free-space wavenumber.

We define a finite-dimensional subspace 
\[
\mathcal H_\delta
=
\operatorname{span}
\left\{
e^{-i\delta\alpha_n u}
\right\}_{n=0}^{M-1}
\subset L^2([-1,1]).
\]

\begin{remark}
For a linear array, the variable $u=\cos\theta$ represents the directional cosine relative to the array axis, so the interval
\([-1,1]\) corresponds to all observation directions. \qed
\end{remark}

\begin{remark}\label{r:spacing}
The element spacing is $x_{n+1}-x_n=d(\alpha_{n+1}-\alpha_n)$. For a uniform linear array, $\alpha_n=n$, so the element spacing $x_{n+1}-x_n=d(n+1)-dn=d$, and for half wavelength separation $d=\frac{\lambda}{2}$, we have $\delta=\pi$. \qed
\end{remark}

\begin{remark}
In antenna-array terminology, each array element contributes one steering function $e^{-i\delta\alpha_n u}$, and the vector
\[
{\bf a}(u)
=
\bigl(
e^{-i\delta\alpha_0 u},
\cdots,
e^{-i\delta\alpha_{M-1} u}
\bigr)^T,
\qquad u\in[-1,1],
\]
is the array steering vector, and the parameterized set
\[
\{{\bf a}(u):u\in[-1,1]\}\subset\mathbb C^M
\]
is commonly referred to as the array steering manifold. The corresponding
array factors are the functions
\begin{equation}\label{eq:af}
f(u)={\bf w}^*{\bf a}(u),
\qquad {\bf w}\in\mathbb C^M,   
\end{equation}
where ${\bf w}\in\mathbb C^M$ consists of given excitation weights, and the superscript $\phantom{ }^*$ denotes the complex conjugate transpose. Thus the set of all possible array factors forms the finite-dimensional exponential subspace $\mathcal H_\delta$. \qed
\end{remark}

\subsection{Array Gain}
The array factor (\ref{eq:af}) is dependent on the spacing $\delta$
\[
f(u)
=
\sum_{n=0}^{M-1}
w_n e^{-i\delta \alpha_n u}\in \mathcal H_\delta,
\quad
u=\cos\theta\in[-1,1].
\]

The maximal array gain in direction \(u=x\) is the extremal quantity
\begin{equation}\label{eq:ag}
\Gamma_\delta(u)
=
\sup_{f\in\mathcal H_\delta\setminus\{0\}}
\frac{|f(u)|^2}{\|f\|_{L^2}^2}.
\end{equation}
where $|f(u)|^2$ is directional radiation intensity, and $\|f\|_{L^2}^2$ is the total radiated power.

\begin{remark}
For the maximal array gain (\ref{eq:ag}), a standard argument with 
$f_1=\frac{f}{\|f\|_{L^2}}$ for any $f\in H_\delta\setminus\{0\}$ shows 
\[
\Gamma_\delta(u)
=
\sup_{\substack{f\in H_\delta\\ \|f\|_{L^2}=1}}
|f(u)|^2. \tag*{\qedsymbol}
\]
\end{remark}

\section{RKHS Interpretation of Array Gain}\label{s:ag}
The following proposition shows that the maximal array gain in any direction $u\in [-1,1]$ can be evaluated as the reproducing kernel at diagonal. 

\begin{proposition}\label{p1}
Let $\mathcal H$ be a reproducing-kernel Hilbert space with kernel $K$.
Then
\[
K(x,x)
=
\sup_{g\in\mathcal H\setminus\{0\}}
\frac{|g(x)|^2}{\|g\|^2}.
\]
\end{proposition}

\begin{proof}
By the reproducing property,
\[
g(x)
=
\langle g,K(\cdot,x)\rangle.
\]
Applying Cauchy--Schwarz inequality, we have
\[
|g(x)|^2
\le
\|g\|^2
\,
\|K(\cdot,x)\|^2.
\]

Since
\[
\|K(\cdot,x)\|^2
=
\langle K(\cdot,x),K(\cdot,x)\rangle
=
K(x,x),
\]
we obtain
\[
\frac{|g(x)|^2}{\|g\|^2}
\le
K(x,x).
\]

Equality is attained by choosing
\[
g=K(\cdot,x).
\]

Therefore,
\[
K(x,x)
=
\sup_{g\neq0}
\frac{|g(x)|^2}{\|g\|^2}. \tag*{\qedsymbol}
\]
\renewcommand{\qedsymbol}{} 
\end{proof}

Since \(\mathcal H_\delta\) is finite-dimensional, it is an RKHS. 
Let $K_\delta$ be the reproducing kernel of $\mathcal H_\delta$. By Proposition \ref{p1} we have
\[
\Gamma_\delta(u)
=
K_\delta(u,u).
\]
Thus, the maximal array gain in the direction $u$ is the diagonal evaluation, i.e., at $(u, u)$, of the reproducing kernel.

In the following, we shall examine the limiting property of $K_\delta$ when the spacing parameter $\delta$ tends to zero.

\section{Spectral Collision Limit}\label{s:scl}
By ``spectral collision" we mean the degeneration of distinct exponential frequencies into repeated frequencies as $\delta\to 0$, i.e., the exponential family that generate $\mathcal H_{\delta}$ undergoes a spectral collision:
\[
\lim_{\delta\to 0}e^{-i\delta\alpha_n u} = 1.
\]

Define a polynomial space
\[
\mathcal P_{M-1}
=
\operatorname{span}
\{1,t,\cdots,t^{M-1}\}\subset L^2([-1,1]).
\]

We have the following spectral collision theorem:

\begin{theorem}\label{t1}
Let $K_{M-1}(\cdot,\cdot)$ be the reproducing kernel of ${\mathcal P_{M-1}}$. Then 
\[
   \lim_{\delta\to 0}
K_\delta (x, y)
= K_{M-1}(x,y)
\]
uniformly for $(x,y)\in [-1,1]\times [-1,1]$.
\end{theorem}

\begin{proof}
Using the Taylor expansion,
\[
e^{-i\delta\alpha_n t}
=
\sum_{k=0}^{M-1}
\frac{(-i\delta\alpha_n t)^k}{k!}
+
O(\delta^M).
\]
Writing it in matrix form we have
\begin{equation}\label{eq1}
e^{-i\delta\boldsymbol\alpha t}
=
VD(\delta){\bf t}
+
\underbar{\bf O}(\delta^M),
\end{equation}
where
\begin{align*}
    \boldsymbol\alpha&=\left(\alpha_0, \cdots, \alpha_{M-1}\right)^T,\\
    V&\in {\mathbb C}^{M\times M} \mbox{ with } V_{n,k}=\alpha_n^k,\quad n, k = 0, 1, \cdots, M-1,\\
D(\delta)&=\operatorname{diag}
\left[
\frac{(-i\delta)^0}{0!}, \frac{-i\delta}{1!},\cdots,\frac{(-i\delta)^{M-1}}{(M-1)!}
\right]\\
&=\operatorname{diag}
\left[ 1, \frac{-i\delta}{1!}\cdots,\frac{(-i\delta)^{M-1}}{(M-1)!}
\right], \\
{\bf t}&=\left(t^0, t, \cdots, t^{M-1}\right)^T=\left(1, t, \cdots, t^{M-1}\right)^T
\end{align*}
and $\underbar{\bf O}(\delta^M)$ is an $M\times 1$ vector with $O(\delta^M)$ entries on $[-1,1]$.

Since the collision parameters $\alpha_n$ are distinct, the Vandermonde matrix $V$ is invertible, allowing recovery of the monomial jet basis from the colliding exponential family. Since $D(\delta)$ is invertible for $\delta\ne 0$, it follows that $V D(\delta)$ is invertible. From (\ref{eq1}) we have
\begin{align*}
{\bf t}&=
D^{-1}(\delta)V^{-1}e^{-i\delta\boldsymbol\alpha t}
+
D^{-1}(\delta)V^{-1}\underbar{\bf O}(\delta^M) \\
&=D^{-1}(\delta)V^{-1}e^{-i\delta\boldsymbol\alpha t}
+
\underbar{\bf O}(\delta). 
\end{align*}

Thus each monomial $t^k$ can be approximated arbitrarily well by vectors in $\mathcal H_\delta$.

Let
\[
{\bf u}^\delta=\left(u_0^\delta,\dots,u_{M-1}^\delta\right)^T=D^{-1}(\delta)V^{-1}e^{-i\delta\boldsymbol\alpha t}.
\]
Then 
\[
u_k^\delta\in \mathcal H_\delta, \quad k=0,\cdots, M-1 
\]
such that
\[
\lim_{\delta\to 0} u_k^\delta = t^k
\]
uniformly for $t\in [-1,1]$.

Define the Gram matrix
\[
G^\delta_{m,n}
=
\langle u_m^\delta,u_n^\delta\rangle.
\]

Since $u_k^\delta\to t^k$, we have
\[
G^\delta\to G,
\]
where
\[
G_{m,n}
=
\langle t^m,t^n\rangle.
\]

Because $\{t^k:k=0,\cdots, M-1\}$ is linearly independent, $G$ is invertible. Hence
\[
(G^\delta)^{-1}\to G^{-1}.
\]

For a finite-dimensional RKHS, the reproducing kernel admits
the representation
\[
K_\delta(x,y)
=
\sum_{m,n=0}^{M-1}
u^\delta_m(x)\,
(G^\delta)^{-1}_{m,n}\,
\overline{u^\delta_n(y)}.
\]
Similarly,
\[
K_{M-1}(x,y)
=
\sum_{m,n=0}^{M-1}
x^m G^{-1}_{m,n} y^n.
\]

Since
\[
u^\delta_k(x)\to x^k
\]
uniformly in $x$, and
\[
(G^\delta)^{-1}\to G^{-1},
\]
each term in the finite sum converges uniformly on
$[-1,1]\times [-1,1]$. Therefore,
\[
K_\delta(x,y)\to K_{M-1}(x,y)
\]
uniformly on $[-1,1]\times [-1,1]$. This completes the proof.
\end{proof}

\begin{remark}
The collision limit above is analogous to the
confluent Vandermonde degenerations arising in Hermite
interpolation theory. Geometrically, the family
\[
\rho \mapsto e^{-i\rho t}
\]
defines an analytic curve in the Hilbert space $L^2([-1,1])$. 
As the spectral points $\rho_n(\delta)=\delta\alpha_n$ coalesce at $\rho=0$, the exponential family loses linear independence at zeroth order, and higher-order jet directions emerge through derivatives:
\[
\partial_\rho^k e^{-i\rho t}\big|_{\rho=0}
=
(-it)^k,
\]
that is, the secant space generated by the exponentials collapses to the osculating jet space generated by the derivatives $(-it)^k$. 
The limiting collision space is precisely the polynomial space $\mathcal P_{M-1}$. \qed
\end{remark}

In the next section we shall present the superdirectivity as a spectral collision limit as $\delta\to 0$ and provide an explanation of superdirectivity through reproducing-kernel geometry.

\section{Superdirectivity as RKHS Boundary Concentration}\label{s:super}
By the spectral-collision theorem of Section \ref{s:scl}, the
scaled exponential family degenerates to the polynomial
jet space $\mathcal P_{M-1}$ as $\delta\to 0$.

In this section, we shall make the connection between the maximal array gain and the Christoffel–Darboux kernel associated with the polynomial space $\mathcal P_{M-1}$, and show that superdirectivity gain of $M^2$ is a phenomenon of its boundary concentration. The Christoffel–Darboux kernel $K_n(\cdot, \cdot)$ is the reproducing kernel for the finite-dimensional Hilbert space of polynomials up to a given degree $n$. 

\begin{theorem}\label{t2}
Let $K_{M-1}(\cdot,\cdot)$ be the Christoffel–Darboux kernel of $\mathcal P_{M-1}$.
Then
$$
\lim_{\delta\to 0}\Gamma_\delta (u)=K_{M-1}(u,u), \quad u\in [-1,1].
$$
In particular, let $\Gamma(u)=K_{M-1}(u,u)$, then
\[
\Gamma(\pm1)=M^2,
\]
and for interior points $|u|<1$,
\[
\Gamma(u)=O(M).
\]
\end{theorem}

\begin{proof}
By (\ref{eq:ag}) and Proposition \ref{p1},
\[
\Gamma_\delta(x)=K_\delta(x,x),
\]
where $K_\delta$ is the reproducing kernel of $H_\delta$.

By Theorem \ref{t1}
\[
K_\delta(x,y)\to K_{M-1}(x,y).
\]

Hence
\[
\Gamma_\delta(x)\to K_{M-1}(x,x).
\]

Define the inner product
\begin{equation}\label{eq:half-ip}
\langle p,q\rangle
=
\frac12
\int_{-1}^{1}
p(t)\overline{q(t)}\,dt.
\end{equation}
Here the interval length is 2 and the measure is normalized appropriately. 

Let
\[
\phi_k(t)
=
\sqrt{2k+1}\,P_k(t),
\qquad k=0,\cdots,M-1,
\]
where $P_k$ denotes the Legendre polynomial of degree $k$.

Since
\[
\frac12
\int_{-1}^{1}
P_k(t)P_j(t)\,dt
=
\frac{\delta_{k,j}}{2k+1},
\]
the functions $\phi_k$ form an orthonormal basis of
$\mathcal P_{M-1}$.

The reproducing kernel is therefore
\begin{equation}\label{eq:KM}
K_{M-1}(s,t)
=
\sum_{k=0}^{M-1}
\phi_k(s)\phi_k(t).    
\end{equation}

Since
$P_k(-1)=(-1)^k$, 
we obtain
\begin{align*}
K_{M-1}(\pm 1,\pm 1)
&=
\sum_{k=0}^{M-1}
(2k+1)P_k^2(-1) \\
&=
\sum_{k=0}^{M-1}
(2k+1)=M^2.
\end{align*}

Furthermore, from (\ref{eq:KM}), we have
\begin{equation*}
K_{M-1}(x,x)
=
\sum_{n=0}^{M-1}
(2n+1)\,P^2_n(x),   
\end{equation*}
where \(P_n\) denotes the Legendre polynomial of degree \(n\). For interior
points \(|x|<1\), the classical asymptotic formula \cite[Theorem 8.21.2]{Szego1975}
\begin{align*}
P_n(\cos\theta)
&=
\sqrt{\frac{2}{\pi n\sin\theta}}
\cos\!\left(
(n+\tfrac12)\theta-\tfrac{\pi}{4}
\right)
+O(n^{-3/2}), \\
&\qquad 0<\theta<\pi,   
\end{align*}
implies
\[
P_n(x)^2 = O(n^{-1}),
\qquad |x|<1.
\]
Consequently,
\[
K_{M-1}(x,x)=O(M),
\qquad |x|<1.
\]
The proof is complete.
\end{proof}

\begin{remark}
Theorem \ref{t2} shows that superdirectivity is fundamentally a boundary amplification phenomenon of reproducing kernels. While the  representation of the gain as a diagonal kernel evaluation is independent of the orthonormal basis chosen to expand the kernel, the specific asymptotic growth rate $M^2$ at the endpoints is a consequence of the underlying flat $L^2([-1,1])$ geometry, not a universal property of all RKHSs. Replacing the Lebesgue measure with a different spectral measure (e.g., a Jacobi weight) would yield a different gain scaling, as discussed in the next section. \qed
\end{remark}

\section{Geometry of Spectral Concentration}\label{s:gsc}
Having established that $\lim_{\delta\to0}\Gamma_\delta(u) = K_{M-1}(u,u)$ and that the limiting kernel yields $M^2$ at the endpoints and $O(M)$ in the interior, we now examine the geometric mechanism underlying this boundary concentration. The Christoffel function provides a variational interpretation linking the array gain to a minimal-energy problem. The trace identity and Hilbert–Schmidt norm reveal how the fixed total gain budget forces energy redistribution from the interior to the hard edge. Finally, we show that, mathematically, the quadratic scaling is not universal but rather a consequence of the flat Lebesgue measure in $L^2([-1,1])$ geometry.

\subsection{Christoffel Function Interpretation}\label{ss:cfi}
The polynomial space $\mathcal P_{M-1}$ links the study of superdirectivity to the Christoffel function with its extensive literature. For the polynomial space $\mathcal{P}_{M-1}$ equipped with the inner product (\ref{eq:half-ip}), the Christoffel function $\lambda_{M}(u)$ is defined as the reciprocal of the diagonal reproducing kernel \cite{Totik2000}:
\[
\lambda_{M}(u) = \frac{1}{K_{M-1}(u,u)}.
\]

Since
\[
\Gamma(u) = K_{M-1}(u,u) = \frac{1}{\lambda_{M}(u)},
\]
a small Christoffel function corresponds to a large array gain, 

For the Lebesgue measure on $[-1,1]$, the orthogonal polynomials are the Legendre polynomials. Classical asymptotics~\cite{Szego1975} yield
\[
\lambda_{M}(u) =
\begin{cases}
O(M^{-1}), & |u| < 1,\\[4pt]
O(M^{-2}), & u = \pm 1.
\end{cases}
\]
The faster decay of the Christoffel function at the endpoints—quadratic rather than linear—quantifies the superdirective gain, i.e., the $M^2$ endfire law is precisely the statement that the Christoffel function collapses as $1/M^2$ at the hard edge of the support, whereas in the interior the collapse is only $1/M$.

A fundamental result in orthogonal polynomial theory~\cite{Totik2000} establishes the equivalent variational characterization
\[
\lambda_{M}(u) = \inf_{\substack{p\in\mathcal{P}_{M-1}\\ p(u)=1}} \|p\|_{L^2}^2.
\]
Physically, $\lambda_{M}(u)$ represents the minimum total radiated power required to achieve a unit response at direction $u$. 

This variational perspective connects array theory to classical extremal problems in approximation theory. The minimizing polynomial achieving the infimum is uniquely given by
\[
p^\star(x) = \frac{K_{M-1}(u,x)}{K_{M-1}(u,u)},
\]
which satisfies $p^\star(u)=1$ and attains the minimum norm $\|p^\star\|_{L^2}^2 = \lambda_{M-}(u)$. The Christoffel function therefore encapsulates the fundamental relationship between directional response and total radiated energy, providing a mathematically rigorous foundation for the superdirective phenomenon.

\subsection{Kernel Concentration and Spectral Energy Redistribution}
The results of Theorem~\ref{t2} has been rephrased in the language of Christoffel functions, with faster endpoint decay quantifying the concentration. We now examine how these large pointwise values are compatible with the fixed total gain budget.

The convergence $K_\delta(x,y) \to K_{M-1}(x,y)$ established in Theorem~\ref{t1} is accompanied by a striking redistribution of the kernel's energy. 
This redistribution occurs under two invariant constraints that hold for any $M$-dimensional RKHS. First, the Hilbert-Schmidt norm of the reproducing kernel satisfies
\[
\iint|K(x,y)|^2 \, d\mu(x) \, d\mu(y) = M.
\]
Second, the trace identity 
$$
\int K(x,x) d\mu(x) = \sum_k \int |\phi_k(x)|^2 d\mu(x) = M.
$$ 
Both hold with respect to the same measure $\mu$ that defines the inner product, and can be verified directly with
$$
K(x,y)=\sum_{k=0}^{M-1}\phi_k(x)\overline{\phi_k(y)}
$$
for an $M$-dimensional RKHS with orthonormal basis $\{\phi_k\}$.

Since $K_\delta$ is a reproducing kernel,
\[
\Gamma_\delta(x)=K_\delta(x,x)
=
\|K_\delta(\cdot,x)\|_{L^2}^2
\ge0.
\]
Thus the diagonal kernel naturally represents a nonnegative
directional amplification density.

The emergence of pointwise values as large as $M^2$ at the endpoints therefore cannot arise from an increase in total kernel mass. Instead, it must correspond to a \emph{concentration} of the existing $L^2$ energy: the kernel withdraws from interior regions and accumulates at the boundary. The trace identity forces the average of $K_{M-1}(x,x)$ over $x$ to equal $M$, so pointwise values of order $M^2$ are possible only if the kernel is correspondingly smaller elsewhere. In the Legendre case, the Christoffel function asymptotics from Subsection \ref{ss:cfi} gives $K_{M-1}(x,x) = O(M)$ for all $|x|<1$ while $M^2$ is achieved only at the corner points, consistent with this redistribution.

This perspective clarifies why superdirectivity is fundamentally a boundary phenomenon. The total ``budget'' of array gain, integrated over all directions, is fixed by the number of array elements:
\[
\int_{-1}^1 \Gamma(u) \, du = \int_{-1}^1 K_{M-1}(u,u) \, du = M.
\]
Achieving gain exceeding $M$ at any direction requires sacrificing gain elsewhere. The spectral collision forces this sacrifice onto the interior, amplifying the endpoints at the expense of broadside directions. The quadratic $M^2$ law is thus not a creation of gain from nothing, but rather a severe redistribution of a fixed total gain budget.

Figs.~\ref{fig1} and~\ref{fig2} illustrate this evolution. For $\delta$ large (coarse spacing), the kernel $|K_\delta(x,y)|$ remains relatively flat across the entire domain $[-1,1]\times [-1,1]$. As $\delta$ decreases, the kernel develops sharp peaks near the corners $(x,y)=(\pm1,\pm1)$, indicating the transfer of spectral weight to the endpoints. Note that the ranges of the color bar are different in Fig. \ref{fig1}, with a much larger range for $\delta=0.1\pi$. Fig. \ref{fig2} shows the evolution of the concentration on the diagonal $K_\delta (u,u)$. Half wavelength spacing corresponds to $\delta=\pi$ (see Remark \ref{r:spacing}).

\begin{figure}[!t]
\centering
\includegraphics[width=4.5in]{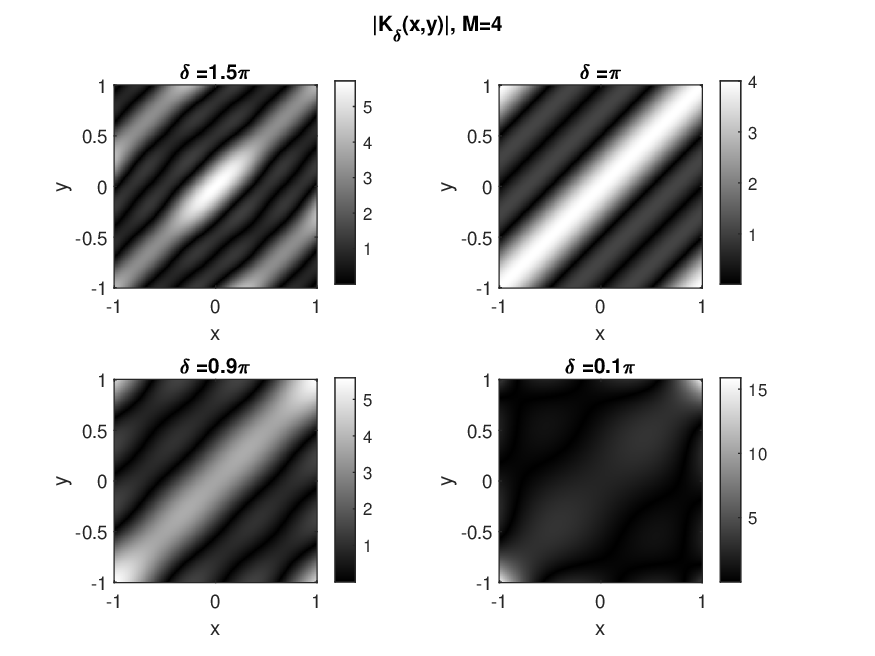}
\caption{Heatmap of $|K_\delta(x,y)|$.}
\label{fig1}
\end{figure}

\begin{figure}[!t]
\centering
\includegraphics[width=4.5in]{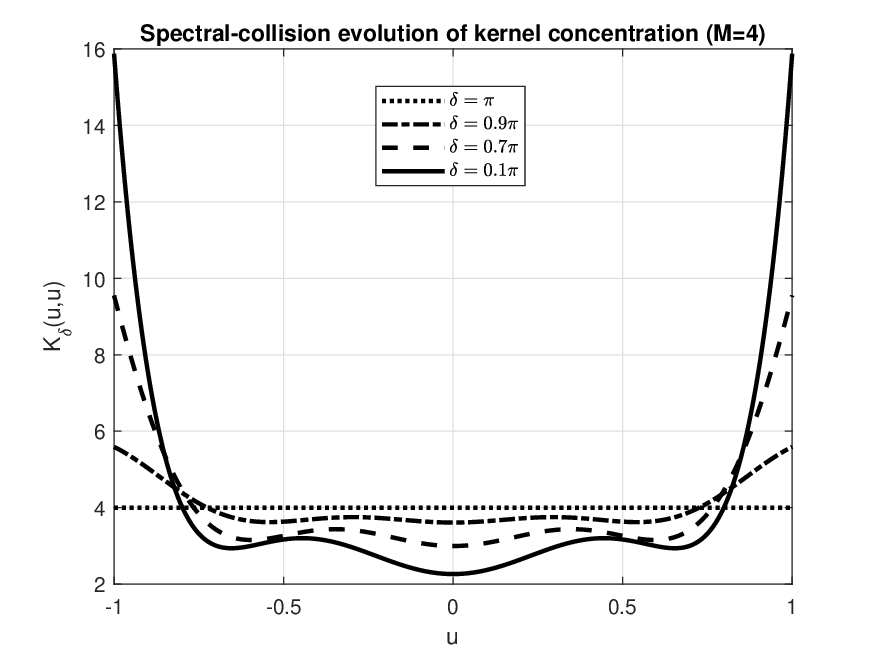}
\caption{Collision evolution of concentration on the diagonal $K_\delta(u,u)$.}
\label{fig2}
\end{figure}

\subsection{Other Spectral Geometries}

The preceding analysis assumes the flat $L^2([-1,1])$ inner product, corresponding to the Lebesgue measure $d\mu(t)=dt$. However, the spectral-collision framework developed in Sections~\ref{s:scl} and~\ref{s:super} accommodates any positive measure $\mu$ on $[-1,1]$ for which the exponentials $e^{-i\delta\alpha_n u}$ lie in $L^2(\mu)$. The collision limit remains the polynomial space $\mathcal{P}_{M-1}$, but now equipped with the inner product
\[
\langle f,g\rangle_\mu = \int_{-1}^1 f(t)\overline{g(t)}\,d\mu(t).
\]
The resulting reproducing kernel $K_{M-1}^{(\mu)}$ determines the asymptotic gain scaling
\[
\Gamma^{(\mu)}(x) = \lim_{\delta\to0}\Gamma_\delta(x) = K_{M-1}^{(\mu)}(x,x).
\]

The classical $M^2$ law is therefore not mathematically universal.  Instead, the growth rate of $K_{M-1}^{(\mu)}(x,x)$ at the endpoints depends sensitively on the local behavior of $\mu$ near $x=\pm1$. For example, for the Jacobi family~\cite{Szego1975}
\[
d\mu^{(\alpha,\beta)}(t) = (1-t)^{\alpha}(1+t)^{\beta}\,dt,\qquad \alpha,\beta > -1,
\]
the orthonormal polynomials are the Jacobi polynomials $P_n^{(\alpha,\beta)}(t)$ normalized appropriately. Classical asymptotics~\cite{Szego1975} yield
\[
K_{M-1}^{(\mu)}(1,1) \sim C(\alpha) \cdot M^{2\alpha+2},
\]
while for interior points $|x|<1$ one retains $K_{M-1}^{(\mu)}(x,x) = O(M)$, provided $\mu$ has a positive continuous density on $(-1,1)$.

Table~1 summarizes the endpoint scaling for representative Jacobi parameters.

\begin{table}[ht]
\centering
\caption{Endpoint gain scaling $K_{M-1}^{(\mu)}(1,1)$ for Jacobi measures $d\mu(t)=(1-t)^{\alpha}dt$ on $[-1,1]$.}
\begin{tabular}{ccl}
\hline
$\alpha$ & Measure behavior at $t=1$ & Gain scaling \\
\hline
$-1/2$ & Chebyshev (first kind), blows up as $(1-t)^{-1/2}$ & $O(M)$ \\
$0$ & Legendre (flat), constant density & $O(M^2)$ \\
$1$ & Vanishes linearly as $(1-t)$ & $O(M^4)$ \\
$2$ & Vanishes quadratically as $(1-t)^2$ & $O(M^6)$ \\
\hline
\end{tabular}
\label{tab:jacobi_scaling}
\end{table}

Thus, Christoffel function asymptotics for Jacobi weights yield endpoint concentration scaling $M^{2\alpha+2}$. For $\alpha > 0$, this would constitute a super concentration effect exceeding the classical $M^2$ superdirective law.

\section{Conclusion}\label{s:concl}
We have shown that array superdirectivity can be understood as a concentration phenomenon induced by spectral collision. For every array spacing, the maximum achievable gain admits an exact characterization as the diagonal of the reproducing kernel associated with the corresponding exponential space. As the element spacing tends to zero, these exponential spaces undergo a confluent degeneration to a polynomial jet space. This spectral-collision limit transforms the reproducing kernel into the Christoffel–Darboux kernel of the limiting polynomial geometry, thereby connecting superdirective gain to Christoffel functions and orthogonal polynomial asymptotics.

For the classical flat \(L^2([-1,1])\) geometry, the limiting kernel is the Christoffel--Darboux kernel of the Legendre polynomials. The well-known quadratic \(M^2\) endfire law follows directly from its endpoint asymptotics and may therefore be interpreted as a manifestation of boundary concentration. From this perspective, superdirectivity is not the creation of additional gain but the redistribution of a fixed gain budget through the concentration of kernel mass near the boundary of the spectral domain.

The reproducing-kernel formulation also clarifies the role of spectral geometry. The classical \(M^2\) scaling is not a universal consequence of dense spacing alone but a consequence of the particular flat geometry associated with the Lebesgue measure on \([-1,1]\). Alternative spectral measures produce different Christoffel asymptotics and therefore different amplification laws. The spectral-collision viewpoint thus places superdirectivity within a broader class of concentration phenomena arising from singular limits of finite-dimensional function spaces and reveals a direct connection between array gain, reproducing kernels, Christoffel functions, and orthogonal polynomial asymptotics.


%




\ifCLASSOPTIONcaptionsoff
  \newpage
\fi



%
\bibliographystyle{IEEEtran} 
\bibliography{references}

%





\end{document}